\documentclass[format=sigconf]{acmart}
\usepackage{algorithm}
\usepackage[utf8]{inputenc}
\usepackage{comment}
  \usepackage{mathtools}
\usepackage{graphicx}
\usepackage{natbib}
\graphicspath{ {./images/} }
\usepackage{tcolorbox}

\newcommand{\probP}{\text{I\kern-0.10em P}}
\newtheorem*{theorem*}{Theorem}

\newtheorem{corollary}{Corollary}
\newtheorem{assumption}{Assumption}

\newtheorem{definition}{Definition}

\newtheorem{lemma}{Lemma}

\newtheorem*{problem*}{Problem}

\DeclareMathOperator*{\argmax}{arg\,max}
\DeclareMathOperator*{\argmin}{arg\,min}
\DeclareMathOperator*{\match}{\leftrightarrow}
\DeclareMathOperator*{\sforall}{\text{ }\forall}

\usepackage[noend]{algpseudocode}
\usepackage{wrapfig,lipsum,booktabs}
\makeatletter
\def\algbackskip{\hskip-\ALG@thistlm}

\usepackage{caption}
\usepackage{todonotes}

\makeatother

\title{Robust Scenario Interpretation from Multi-model Prediction Efforts}
\author{Yuanhao Lu}
\affiliation{University of Southern California\\Los Angeles
\country{USA}}
\email{terryl@usc.edu}

\author{Ajitesh Srivastava}
\affiliation{University of Southern California\\Los Angeles
\country{USA}}
\email{ajiteshs@usc.edu}
\acmConference[SIGKDD'22]{August 14--18, 2022 }{Washington, DC, USA}

\date{August, 2022}

\begin{document}

\begin{abstract}
    Multi-model prediction efforts in infectious disease modeling and climate modeling involve multiple teams independently producing projections under various scenarios. Often these scenarios are produced by the presence and absence of a decision in the future, e.g., no vaccinations (scenario A) vs vaccinations (scenario B) available in the future. The models submit probabilistic projections for each of the scenarios. Obtaining a confidence interval on the impact of the decision (e.g., number of deaths averted) is important for decision making. However, obtaining tight bounds only from the probabilistic projections for the individual scenarios is difficult, as the joint probability is not known. Further, the models may not be able to generate the joint probability distribution due to various reasons including the need to rewrite simulations, and storage and transfer requirements.
    Without asking the submitting models for additional work, we aim to estimate a non-trivial bound on the outcomes due to the decision variable. We first prove, under a key assumption, that an $\alpha-$confidence interval on the difference of scenario predictions can be obtained given only the quantiles of the predictions. Then we show how to estimate a confidence interval after relaxing that assumption. We use our approach to estimate confidence intervals on reduction in cases, deaths, and hospitalizations due to vaccinations based on model submissions to the US Scenario Modeling Hub.
\end{abstract}

\maketitle
\section{Introduction}

To leverage the wisdom of multiple experts in predictions, various fields employ the approach of coordinating multiple teams who independently submit their projections. Such multi-model prediction efforts are common in infectious disease modeling~\cite{borchering2021modeling,FluSight,sharma_dev_mangla_wadhwa_mohanty_kakkar_2021} and climate modeling~\cite{tebaldi2007use,tegegne_melesse_2020,najafi_robertson_massah_2021}. Often, multiple projections are performed under various ``scenarios'' produced by the presence and absence of a decision in the future, e.g., no vaccinations (scenario A) vs vaccinations (scenario B) available in the future. For example, the US/CDC COVID-19 Scenario Modeling Hub coordinates the task of long-term public health impacts under different scenarios \cite{covid_scenario_hub}. The participating models \cite{sikjalpha,bucky,CovidSP,mobs} in this effort predict COVID-19 cases, hospitalization, and deaths for a scenario in weeks ahead as a random variable and output 23 quantiles for each prediction. 

The goal of generating scenario projections based on a future decision is to assess the impact of that decision. It is crucial to identify confidence intervals on the impact. However, this is difficult mainly because the two sets of scenario projections $A$ and $B$ are independently generated -- the joint distribution of any given outcome $X$(e.g., number of deaths on a certain date) under the two scenarios is not known to enable computing $P(X|B - X|A)$. Also, the distribution is available as a set of quantiles rather than a continuous cumulative distribution function (CDF). Asking the models to generate a joint distribution of outcomes adds to the challenges. First, the modeling teams already spend a significant amount of time on modeling and projections. Additional work of computing the joint distribution, every time a new decision is to be evaluated, may create barriers to joining the multi-model effort and releasing timely projections. Further, submitting joint distributions $P(X|A, X|B)$ will quadratically increase the space complexity. Currently, the submissions in US Scenario Modeling Hub with only marginal distributions $P(X|A)$ for four scenarios for 50 states and 52 week-long projections and three targets (cases, deaths, and hospitalizations) result in a file size over 100MB.

Our goal is to identify non-trivial bounds on the difference of outcomes under two scenarios without asking for any additional work from the participating modeling teams. We assume that we have quantiles for targets for at least two scenarios of which the difference is to be computed. We make some realistic assumptions supported by observed data to derive an arbitrary $\alpha$-confidence interval for the difference of two scenarios.
We start with a strong assumption (Assumption~\ref{Assumption1}) on the models that result in what we define as \textit{zero-violation models}. Under this assumption, we devise a method to find arbitrary $\alpha$-confidence. We then relax assumption 1 and define \textit{$\epsilon-$violation as model} as a \textit{model} that partially satisfies Assumption~\ref{Assumption1} and claim any arbitrary $\alpha$-confidence interval on scenario difference can be still obtained if the \textit{$\epsilon$-violation model} follows Assumption~\ref{Assumption2}. We demonstrate indications that the assumption is reasonable through experiments on models submitted to the US Scenario Modeling Hub.

\section{Methodology}
In this section, we describe a methodology to bound the difference in scenarios for any arbitrary confidence interval $\alpha$ under some assumptions for the models. It should be noted that methods proposed in this section apply to different scenario modeling problems where Assumption \ref{Assumption1} holds. In this paper, we focus on the context of COVID-19 multi-model predictions.

\subsection{Problem Setting}
Define a \textit{scenario} $S_{t_{\text{app}}} \in \mathcal{S}$ to be an environment in which some events occurred at time $t_{\text{app}}$. For instance, a scenario could be recommending vaccination for children at 2 weeks ahead of the time of prediction, so $t_{\text{app}} = 2 + t_0$, where $t_0$ is the time when the prediction is made. Notice that it is possible that $t_0 \neq t_{\text{app}}$, \textit{i.e} events that distinguish scenarios need not immediately occur at the time of prediction. Within each scenario, each unique set of latent variables $P_0, P_1, ..., P_n$ and the time being predicted $t$ is associated with a separate result (COVID-19 case, hospitalization, deaths). For instance, if we are predicting COVID-19 cases, latent variables $P_0, P_1, ..., P_n$ could be temperature, human mobility, percentage of masked population, \textit{etc}. 

A \textit{model} in the multi-model effort takes in a future time $t \geq t_0$ and a \textit{scenario} $S_{t_\text{app}}$ as inputs, and outputs a stochastic prediction of the COVID-19 cases, hospitalizations, and deaths associated with the scenario. We can interpret the \textit{model} that predicts the (COVID cases, deaths, or hospitalizations) for 2 scenarios as a random vector constructed by repeatedly uniformly sampling form the at most countably infinite universe $U(t) = \{(x_i,y_j)\}_N$, where $U$ is dependent the time $t$ that we are predicting for. Each $x_i,y_j \in \mathbb{R}$ is one possible prediction for the two scenarios being compared and is affected by the latent variables $P_1,P_2,...,P_n$ as indicated in figure \ref{fig:tree}. For \textit{models} capable of predicting more than 2 scenarios, interpret them as uniformly sampling the random vectors of length $|\mathcal{S}|$, where $|\mathcal{S}|$ is the number of possible scenarios. Then, for each vector, truncate it by selecting only elements $x_i$ and $y_i$ to form the 2-vector $(x_i,y_j)$, corresponding to the two scenarios of interest. This interpretation of \textit{models} is equivalent to a Monte-Carlo simulation with the underlying latent variables. Figure \ref{fig:tree} provides a detailed illustration of the interpretation of a \textit{model}.

Now, label $x_i$ and $y_j$ such that $x_0 \leq x_1, ..., \leq x_N$ and $y_0 \leq y_1, ..., \leq y_N$, and let $X$, $Y$ be the respective random variables obtained by sampling $x_i$ and $y_j$ separately and independently from $U = \{(x_i,y_j)\}_N$. In our multi-model effort, the \textit{model} only provides information for $X$ and $Y$ in the form of 23 distinctive quantiles, denote them as $\mathbf{Q}^{X} = \{Q^{X}_{q_1}, Q^{X}_{q_2}, ..., Q^{X}_{q_{23}}\}$ and $\mathbf{Q}^{Y} = \{Q^{Y}_{q_1}, Q^{Y}_{q_2}, ..., Q^{Y}_{q_{23}}\}$ where $\{q_1, q_2, ..., q_{23}\} = \mathbf{q}$. Also, let $F_X(k)$ and $F_Y(k)$ denote the cumulative distribution function of random variable $X$ and $Y$ evaluate at $k$, respectively. 
\begin{definition}
The tuple $(x_i,y_j)$ is a matching if and only if $(x_i,y_j) \in U(t)$ for some $t$. Denote this matching by $x_i \match y_j$.
\end{definition}

At the right of the figure \ref{fig:tree}, the elements $x_i$ and $y_j$ with the same color represent a \textit{matching} pair $(x_i,y_j)$. Observe that $x_i$ and $y_j$ having the same set of latent variables $(P_1,P_2,...,P_n,t)$ is a necessary but not a sufficient condition for $x_i \match y_j$.  

From this interpretation, our problem of bounding difference of scenarios can be formulated as the following:
\begin{tcolorbox}
\textit{Given $\alpha \in (0,1)$ and $t \in \mathbb{N}$, give an $\alpha$-confidence interval for $Z$, the random variable obtained from uniformly sampling $x_i - y_j$ where $x_i \match y_j \in U(t)$}
\end{tcolorbox}
To do this, we need to extract the \textit{matching} information from the 23 quantiles given. We thereby propose the following assumption on the scenario:

\begin{assumption}\label{Assumption1}
(Monotonic Impact of Latent Variables): Changing any subset of latent variables will impact the cases, hospitalizations, and deaths of all scenarios in the same direction.
\end{assumption}

This assumption also follows the real-life observations of epidemics. Usually, the latent factors include more transmissible variants, lack of medical resources, social mobility, population awareness, etc, combining these factors should give us monotonic impact on the potency of the epidemic. It should be noted that, however, this assumption is strong, and we will discuss the limitation of this assumption and relax it in the part where \textit{violation models} are discussed.

\subsection{Zero-Violation Models}

For any \textit{models} predicting scenarios that satisfy assumption \ref{Assumption1}, we would expect it to demonstrate a similar behavior in its predictions. That is, we can reasonably expect for all chosen pair $(x_i,y_j) \in U(t)$ from the \textit{model}, changing the latent variables $P_0, ..., P_n$ will only impact $x_i$ and $y_j$ in the same direction. In other words, any combination of interventions that could have reduced $x_i$ could not increase $y_i$, prompting the following definition:

\begin{definition}
A zero-violation \textit{model} is a model that satisfies assumption \ref{Assumption1}.
\end{definition}

Due to the stochasticity of some models, the leaf nodes in figure \ref{fig:tree} are not necessarily deterministic but are rather random vectors $(X_i,Y_i)$. Taking account of this type of model and reconciling for scenarios where assumption \ref{Assumption1} does not hold, \textit{violation models} is defined and discussed in later sections. For now, keep in mind that \textit{zero-violation models} have deterministic leaf nodes in figure \ref{fig:tree}.

\begin{lemma}\label{lem:1}
(Well-Orderedness of \textit{zero-violation models}): For $x_i \match y_j$ produced by a \textit{zero-violation model}, we must have $i=j$. That is, the rank of \textit{matching} $x_i$ and $y_j$ must be equal in their respective ordered list. 
\end{lemma}

\begin{proof}
The proof can be found in Appendix \ref{lem:1proof}
\end{proof}

Lemma \ref{lem:1} implies that $F_X(X \leq x_i) = F_Y(Y \leq y_i) = \frac{i}{N}$ for any \textit{zero-violation model}. Therefore, to sample $Z = X - Y$, it is sufficient to choose $i$ uniformly at random from $\{1,2,...N\}$ and calculate $x_i - y_i$. In our problem, however, only 23 $x_i$ and $y_i$ are given with their corresponding ranks (provided by the 23 quantiles). To adapt, we propose algorithm \ref{alg:IIMB} to find the set of possible upper and lower bound on $Z$, from which any arbitrary $\alpha$-confidence interval on $Z$ can be obtained.

\begin{algorithm}
\begin{flushleft}
\textbf{Input}: Quantile labels $\mathbf{q}$, quantiles $\mathbf{Q^{X_t}}$ and $\mathbf{Q^{Y_t}}$ \\ 
\textbf{Output}: Upper and lower bound on $Z = x_i - y_j$, $Z^U$ and $Z^L$, at time $t$ 
\end{flushleft}

\caption{Iterative Zero-Violation Model Bound}\label{alg:IIMB}
    \begin{algorithmic}[1]
        \For{$j = \{1,2,...,10^5\}$}
        \State{Sample $\frac{i}{N}$ from Uniform$(0,1)$}
        \State{$q_l \gets \max_{q \in \mathbf{q}}\{q \leq \frac{i}{N}\}$}
        \State{$q_u \gets \min_{q \in \mathbf{q}}\{q \geq \frac{i}{N}\}$}
        \State{Append $\mathbf{Q^X}_{q_u} - \mathbf{Q^Y}_{q_l}$ to $Z^U$}
        \State{Append $\mathbf{Q^X}_{q_l} - \mathbf{Q^Y}_{q_u}$ to $Z^L$}
        \EndFor
    \end{algorithmic}
\end{algorithm}

In Algorithm\ref{alg:IIMB}, since $Q^X_{q_l} \leq x_i \leq Q^X_{q_u}$ and $Q^Y_{q_l} \leq y_i \leq Q^Y_{q_u}$, it is easy to see that 
\begin{align}
    Q^X_{q_l} - Q^Y_{q_u} \leq x_i &- y_i \leq Q^X_{q_u} - Q^Y_{q_l}
    \intertext{that is,}
    Z^L \leq &Z \leq Z^U
\end{align} 

Since $\probP(l\leq Z\leq u) = \probP(Z \leq u) - \probP(Z \leq l)$, and $\probP(Z \leq u) \geq \probP(Z^U \leq u)$, $\probP(Z\leq l) \leq \probP(Z^L \leq l)$, we arrive at the identity $\probP(l \leq Z \leq u) \geq P(Z^{U} \leq u) - \probP(Z^L \leq l)$, which can be applied to find the confidence interval. Namely, an $\alpha$-confidence interval is obtained by finding the appropriate $u$ and $l$ such that $\probP(l \leq Z \leq u) \leq \probP(Z^{U} \leq u) - \probP(Z^L \leq l) = \alpha$. 

In addition, as the number of available quantiles increases uniformly on $[0,1]$, $q_l$ converges to $q_u$ ($q_l \to q_u$). Thus, $Q^X_{q_l} \to Q^X_{q_u}$, $Q^Y_{q_l} \to Q^Y_{q_u}$, and $Z^U - Z^L = Q^X_{q_u} - Q^X_{q_l} + Q^Y_{q_l} - Q^Y_{q_u} \to 0$; that is, $Z^L \to Z^U = Z$ eventually when large number of quantiles become available. 
\begin{figure*}
    \begin{center}
        \includegraphics[width=0.8\linewidth]{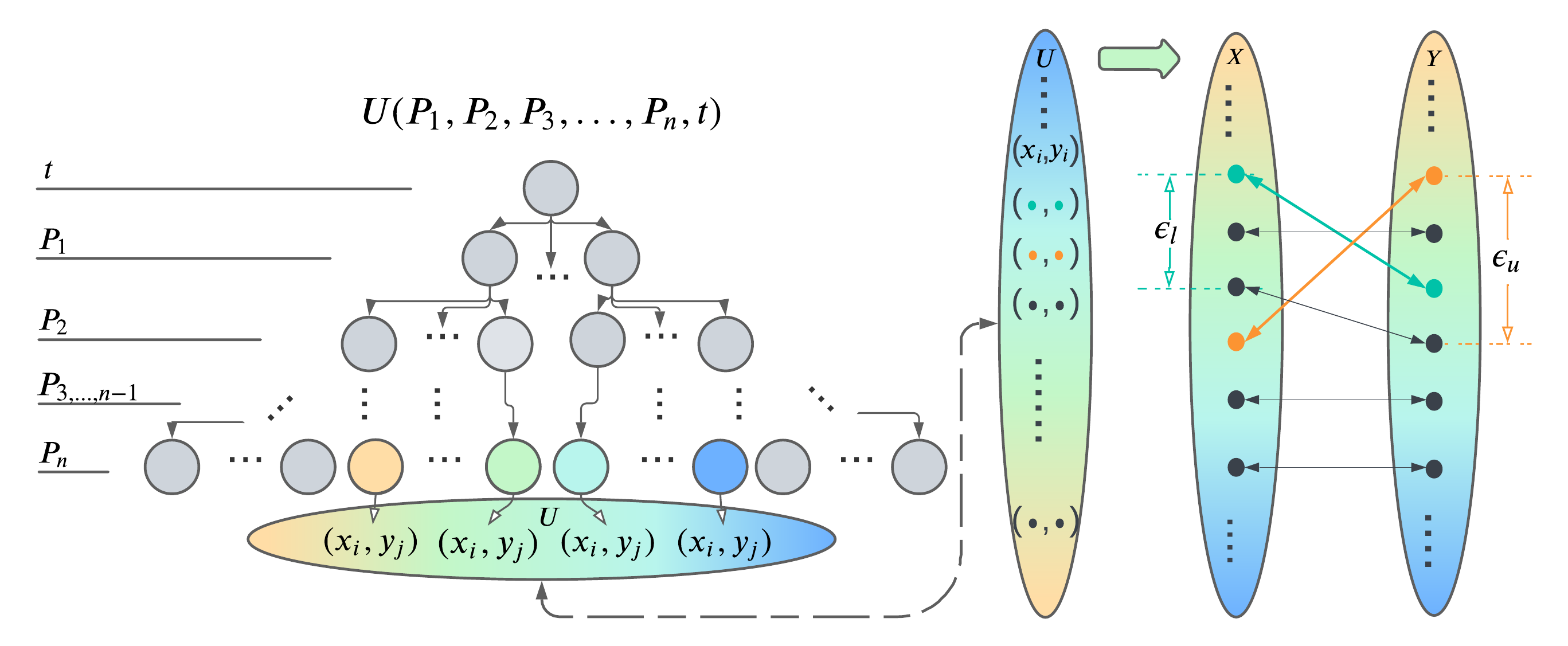}
    \end{center}    
    \caption{Illustration for the interpretation of the \textit{zero-violation model} at time $t$. The root node represents the time $t$ that the \textit{model} is generating predictions for, while every other non-leaf nodes are latent variables $P_i$ affecting predictions. As it might not be feasible for the \textit{model} to capture the precise dynamics of the environment(scenario), $P_i$ in the \textit{model} can be different from the $P_i$ of the scenario. The leaf nodes are vectors in $\mathbb{R}^{|\mathcal{S}|}$, where $|\mathcal{S}|$ is the total number of scenarios the model predicts (in this figure, $|\mathcal{S}| = 2$). The leaf nodes are affected by the variables from all of its ancestor nodes. Each unique path in this tree generates a distinct vector, and the collection of which forms the universe at time $t$. The two lists $X$ and $Y$ depicted at the right of $U$ are the ranked independent  observations of $x_i$ and $y_j$. }
    \label{fig:tree}
\end{figure*}

\subsection{Violation Models}
A \textit{violation model} is a \textit{model} that only partially satisfies assumption \ref{Assumption1}. From Lemma \ref{lem:1}, a \textit{zero-violation model} is expected to have $x_i \match y_j, \sforall i = j$. Equivalently, for \textit{zero-violation models} with scenarios whose $t_{\text{app}} > t_0$, $x_i = y_i$ is expected because $x_i \match y_i$ and the events that differentiates the scenarios has not occurred at time $t \leq t_{\text{app}}$, which entails $Q^X_{q_i} = Q^Y_{q_i} \sforall q_i \in \mathbf{q}$. However, this is not observed in practice for two reasons:
\begin{enumerate}
    \item Some models have stochasticity, and the leaf nodes in figure \ref{fig:tree} becomes random variables $(X_i,Y_j)$ instead of the deterministic vector $(x_i,y_j)$. From repeated sampling $x_i$ and $y_j$ from $X_i$ and $Y_j$, we can not guarantee $x_i = y_i$ for $t \leq t_{\text{app}}.$  
    \item Assumption 1 is violated in the models. This can happen in many ways. For instance, vaccines can be more potent under some specific circumstances. 
\end{enumerate}

In reality, a weaker version of assumption \ref{Assumption1} is more likely to hold: cases, hospitalizations, and deaths are monotonic with changing one latent variable, instead of changing a set of latent variables. That is, each vaccination, human mobility, weather, etc. has a monotonic impact, but the monotonic impact is not guaranteed when the changes are combined. Empirically, however, we later demonstrate that the extent to which the assumption is violated is small.

For \textit{models} that violates assumption \ref{Assumption1} or have stochasticity (random vectors as leaf nodes in figure \ref{fig:tree}), we define them as \textit{violation models}. To quantitatively examine this type of \textit{models}, define \textit{$\epsilon$-violation} as a measure on the degree of violation for \textit{violation models}.

\subsubsection{$\epsilon$-Violation}
A mismatch occurs when $x_i \match y_j$ with $i \neq j$. There are two kinds of mismatch that can happen: either $i > j$ or $i < j$. In other words, $x_i$ is matched to an element that is above its rank in $Y$'s ordered list or vice versa. The right hand side of Figure \ref{fig:tree} illustrates this: depicted in orange, $\epsilon_u$ is the maximum upper mismatch of $x_i \match y_i$, while $\epsilon_l$ is the maximum lower mismatch depicted in green. To bound $x_i - y_j$ for $x_i \match y_j$ of a violation model, we therefore need to take the maximum of $x_i - x_{i + \text{mismatch}}$ and $x_i - x_{i - \text{mismatch}}$. The following definition formally defines mismatch.

\begin{definition}\label{def:ep}
The $\epsilon$-violation of a \textit{violation} model is the largest difference in quantiles for $x_i \match y_j$ in their respective ordered lists, \textit{i.e}
\begin{align}
    \epsilon_u &= \max_{x_i \match y_j,t}\{F_Y(x_i) - F_X(y_j)\}\\
    \epsilon_l &= -\min_{x_i \match y_j,t}\{F_Y(x_i) - F_X(y_j)\}\
\end{align}
where $t$ in max and min represents all predictions over time. 
\end{definition}

Intuitively, the larger the $\epsilon_l$ and $\epsilon_u$, the higher the mismatch, and thus more uncertainty the model's predictions. Also, since $x_i \match y_j \Rightarrow i = j$ for \textit{zero-violation models}, we have $F_X(x_i) = F_X(x_j) = F_Y(y_j)$ for all $x_i \match y_j$ for \textit{zero-violation models}, prompting the following corollary:

\begin{corollary}   
A model has $\epsilon_u = \epsilon_l = 0$ if and only if it is a \textit{zero-violation} model. 
\end{corollary}

To bound $\{x_i - y_j:x_i\match y_j\}$ for \textit{violation models} with algorithm \ref{alg:IIMB}, we need to account for the cases where $y_j$ exceeds $Q^Y_{q_u}$ or falls below $Q^Y_{q_l}$, as shown in the right hand side of figure \ref{fig:tree}. Similar adjustments need to be made for $x_i$ as well, prompting the need for wider bound. To do this, instead of sampling $q_l$ and $q_u$ such that $q_l \leq \frac{i}{N} \leq q_u$, sample $q_l$ and $q_u$ to be the tightest quantiles such that $q_l \leq \max\{\frac{i}{N} - \epsilon\ , \min\{\mathbf{q}\}\}$ and $q_u \geq \min\{\frac{i}{N} + \epsilon, \max\{\mathbf{q}\}\}$, where taking the $\max$ and $\min$ with $\mathbf{q}$ enforces $q_u$ and $q_l$ to be within $[\min\{\mathbf{q}\},\max\{\mathbf{q}\}]$, the range of available quantiles. Sampling this way guarantees the $y_j$ which $x_i$ is supposed to be \textit{matched} with lies within $[Q^Y_{q_l}$, $Q^Y_{q_u}]$, and the $x_k$ that $y_j$ is supposed to be \textit{matched} with is in the interval $[Q^X_{q_l},Q^X_{q_u}]$. Then, similar to the \textit{zero-violation models}, we attain the confidence interval by finding the appropriate $u$ and $l$ such that $\probP(l \leq Z \leq u) \geq \probP(Z^{U} \leq u) - \probP(Z^L \leq l)$. Algorithm \ref{alg:ISMB} illustrates the sampling of $Z^U$ and $Z^L$ in detail.

\begin{algorithm}
\begin{flushleft}
\textbf{Input}: Violation parameters $\epsilon_u$ and $\epsilon_l$, quantile labels $\mathbf{q}$, quantiles $\mathbf{Q^X}$ and $\mathbf{Q^Y}$ \\ 
\textbf{Output}: Upper and lower bounds $Z^U$ and $Z^L$ at time $t$ 
\end{flushleft}
\caption{Iterative Violation Model Bound}\label{alg:ISMB}
    \begin{algorithmic}[1]
        \For{$j = \{1,2,...,10^5\}$}
        \State{Sample $\frac{i}{N}$ from Uniform$(0,1)$}
        \State{$q_l \gets \max_{q \in \mathbf{q}}\{q \leq \max\{\frac{i}{N} - \epsilon_l,\min\{\mathbf{q}\}\}\}$}
        \State{$q_u \gets \min_{q \in \mathbf{q}}\{q \geq \min\{\frac{i}{N} + \epsilon_u,\max\{\mathbf{q}\}\}\}$}
        \State{Append $\mathbf{Q^X}_{q_u} - \mathbf{Q^Y}_{q_l}$ to $Z^U_t$}
        \State{Append $\mathbf{Q^X}_{q_l} - \mathbf{Q^Y}_{q_u}$ to $Z^L_t$}
        \EndFor
    \end{algorithmic}
\end{algorithm}
In algorithm \ref{alg:ISMB}, since the CDF is a monotonically increasing function, and $\epsilon_u$ and $\epsilon_l$ is strictly positive,
\begin{equation}
    x_i - F^{-1}_Y(\frac{i}{N} + \epsilon_u) \leq x_i - y_j \leq x_i - F^{-1}_Y(\frac{i}{N} - \epsilon_l)
\end{equation}
by the definition of $\epsilon_u$ and $\epsilon_l$ for $\forall x_i \match y_j$. Then, since $Q^Y_{q_u}$ and $Q^Y_{q_l}$ are sampled such that $Q^Y_{q_u} \geq F^{-1}_Y(\frac{i}{N} + \epsilon_u)$ and $Q^Y_{q_l} \leq F^{-1}_Y(\frac{i}{N} - \epsilon_l)$, we have
\begin{align}
x_i - Q^Y_{q_u} &\leq x_i - F^{-1}_{Y}(\frac{i}{N} + \epsilon_u) \\
&\leq x_i - y_j \\ 
&\leq x_i - F^{-1}_Y(\frac{i}{N} - \epsilon_l) \\
&\leq x_i - Q^Y_{q_l}
\intertext{since $Q^X_{q_l} \leq x_i \leq Q^X_{q_u}$,}
 Q^X_{q_l} - Q^Y_{q_u} \leq x_i &- y_j \leq Q^X_{q_u} - Q^Y_{q_u}
\intertext{that is,}
Z^L \leq &Z \leq Z^U
\end{align}

The task remains to extract $\epsilon$ from the quantiles. Since only 23 quantiles and their respective \textit{matchings} are known, it is not possible to obtain information on $F_X(x_i) - F_Y(y_j)$ for the full distribution. Therefore, $\epsilon$ can be only estimated with the \textit{matchings} $x_i \match y_j$ observable from the quantiles. As the \textit{matchings} are only known  for scenarios with $t_{\text{app}} \geq t_0$ ($Q_i^X$ should be equal to $Q_i^Y$ since scenario has not taken effect. If not, all misalignment are due to stochasticity), only those scenarios are valid for the approximation of $\epsilon$. Now, we propose the concrete method to estimate such $\epsilon$ from scenarios with $t_{\text{app}} \geq t_0$.

\subsubsection{Estimating $\epsilon$ for Violation Models}
Due to the above-mentioned reasons, additional assumptions on the behavior of the violation measure need to be proposed in order to bound the difference in scenarios with $t > t_{\text{app}}$. 

\begin{assumption}\label{Assumption2}
A \textit{violation model} is said to be \textit{well-behaved} if its outputs satisfies both
\begin{align}
    \argmax_{t}\{F_Y(y_j) - F_X(x_i)\} \leq t_{\text{app}}\\
    \intertext{and}
    \argmin_{t}\{F_Y(y_j) - F_X(x_i)\} \leq t_{\text{app}}
\end{align}
for $\forall x_i \match y_j$ \hfill (Non-increasing $\epsilon$ after $t_{\text{app}}$)
\end{assumption}

As discussed in the problem setting, both $F_Y(y_i)$ and $F_X(x_i)$ are sampled from $U(t)$ and is dependent on $t$. Here, with a slight abuse of notation, we take $\argmax$ and $\argmin$ on $t$ to restrict the behavior of the violation across time. Intuitively, assumption \ref{Assumption2} restricts the time at which maximum violation occurs to be before $t_{\text{app}}$, so that the maximum violation is observable. In effect, this assumption enforces the upper-bound estimated on violation for $t \leq t_{\text{app}}$ is still an upper bound for $t > t_{\text{app}}$.

\begin{definition}
Without loss of generality, let $X$ be the scenario such that $Q^X_{q_i} \geq Q^Y_{q_i}$, depending on $i$. Define the estimated upper and lower violation $\tilde{\epsilon_l}$ and $\tilde{\epsilon_u}$ for a \textit{violation model} as the following:

\begin{equation}
    \tilde{\epsilon}_l = \max_{i}\{q_i - q_\alpha\}
    \text{, where }
    \alpha = \max_{k\leq i-1}{\{Q^X_{q_k} \ge Q^Y_{q_{i-1}}\}}
\end{equation}
and
\begin{equation}
    \tilde{\epsilon}_u = -\min_{i}\{q_i - q_\beta\}
    \text{, where }
    \beta = \min_{k\geq i+1}{\{Q^Y_{q_k} \ge Q^X_{q_{i+1}}\}} 
\end{equation}
for all $Q^X$ and $Q^Y$ corresponding to $t < t_{\text{app}}$
\end{definition}

\begin{lemma}\label{lem:2}
For the estimated upper and lower violation $\epsilon_u$ and $\epsilon_l$ and the approximated violation $\tilde{\epsilon}_u$ and $\tilde{\epsilon}_l$, we have
\begin{equation}
    \epsilon_u \leq \tilde{\epsilon}_u \text{ and } \epsilon_l \leq \tilde{\epsilon}_l
\end{equation}
\end{lemma}

\begin{proof}
The proof can be found in Appendix \ref{lem:2proof} 
\end{proof}

The implementation detail of estimating $\tilde{\epsilon}_u$ and $\tilde{\epsilon}_l$ is given in algorithm \ref{alg:estimatingeuel}. The estimated $\epsilon$ obtained can be then plugged into algorithm \ref{alg:IIMB} to evaluate the $\alpha$-confidence interval.

\begin{algorithm}
\caption{Estimating $\epsilon_l$ and $\epsilon_u$}\label{alg:estimatingeuel}
\begin{flushleft}
\textbf{Input:} Quantile labels $\mathbf{q}$, quantiles $\mathbf{Q^{X_t}}$ and $\mathbf{Q^{Y_t}}$, $t_{\text{app}}$\\
\textbf{Output}: $\tilde{\epsilon}_u$ and $\tilde{\epsilon}_l$
\end{flushleft}

\begin{algorithmic}[1]
    \State {$\mathbf{\epsilon}_u \gets \text{[ ]}$}
    \State {$\mathbf{\epsilon}_l \gets \text{[ ]}$}
    \For{$t = 1,2,...,t_{app}$}
        \For{$i = 1,2,...,|\mathbf{q}|$}
            \If{$\mathbf{Q^{X_t}}[i] \geq \mathbf{Q^{Y_t}}[i]$}
                \State{$\mathbf{Q^{U_t}} \gets \mathbf{Q^{X_t}}$}
                \State{$\mathbf{Q^{L_t}} \gets \mathbf{Q^{Y_t}}$}
            \Else
                \State{$\mathbf{Q^{U_t}} \gets \mathbf{Q^{Y_t}}$}
                \State{$\mathbf{Q^{L_t}} \gets \mathbf{Q^{X_t}}$}            
            \EndIf
            \State{$\alpha \gets \max_{k\leq i}\{\mathbf{Q^{U_t}}[k] \geq \mathbf{Q^{L_t}}[i-1]\}$}
            \State{$\beta \gets \min_{k \geq i}\{\mathbf{Q^{L_t}}[k] \geq \mathbf{Q^{U_t}}[i+1]\}$}
            \State{Append $\mathbf{q}[i] - \mathbf{q}[\alpha]$ to $\mathbf{\epsilon}_u$}
            \State{Append $\mathbf{q}[\beta] - \mathbf{q}[i]$ to $\mathbf{\epsilon}_l$}
        \EndFor
    \EndFor
    \State{$\tilde{\epsilon}_l \gets \max\{\mathbf{\epsilon}_l\}$}
    \State{$\tilde{\epsilon}_u \gets \max\{\mathbf{\epsilon}_u\}$}
\end{algorithmic}
\end{algorithm}

\subsubsection{Approximating $\tilde{\epsilon}$}
As shown in \ref{lem:2}, both $\tilde{\epsilon}_u$ and $\tilde{\epsilon}_l$ are guaranteed to be over-estimations for $\epsilon_u$ and $\epsilon_l$.

In the scenarios where the $x_i$ and $y_j$ are not clustered around certain values, their cumulative distribution functions should be relatively "smooth". As the estimated $\epsilon$ is always maximized in the worst-case scenario (all the values are clustered so that the CDF looks stair-like), it is reasonable to approximate the estimation of $\epsilon$ to obtain a smaller uncertainty bound. Naturally, the CDF of $X$ and $Y$ can be interpolated to obtain more quantiles. The monotonicity of Cubic Hermite Interpolating Polynomial (PCHIP) makes it a suitable candidate for interpolating CDF. As stated earlier, the smaller the $\epsilon$, the closer $Z^{L}$ is to $Z^{U}$ in algorithm \ref{alg:IIMB}, resulting in a tighter confidence interval.

\begin{definition}
Depending on the choice of $i$, let $X$ be the scenario associated with a higher quantile $(Q^X(i) > Q^Y(i))$. let $\textsc{PCHIP}_X(x_i):X \to q \in [0,1]$ be the interpolated CDF of scenario X (i.e $F_X$). Define the approximated $\tilde{\epsilon}_l$ and $\tilde{\epsilon}_u$ as
\begin{align}
    \varepsilon_u &= \max_{q}\{\text{PCHIP}_Y(x) - \text{PCHIP}_X(x)\}\\
    \varepsilon_l &= -\min_{q}\{\text{PCHIP}_Y(x) - \text{PCHIP}_X(x)\}
\end{align}
for $\forall x \in [\min\{\mathbf{Q^Y},\mathbf{Q^X}\},\max\{\mathbf{Q^X},\mathbf{Q^Y}\}]$.
\end{definition}

As more quantiles become available, the PCHIP interpolation in \ref{alg:approximateep} becomes a more accurate approximation of the inverse CDF. Since PCHIP is always a refined approximation of the quantiles and combining lemma \ref{lem:2}, the relations $\varepsilon_u \leq \tilde{\epsilon}_u \leq \epsilon_u$ and $\varepsilon_l \leq \tilde{\epsilon}_l \leq \epsilon_l$ hold.

\begin{algorithm}
\caption{Approximating $\tilde{\epsilon}$ with PCHIP}\label{alg:approximateep}
\begin{flushleft}
\textbf{Input}: Quantile labels $\mathbf{q}$, quantiles $\mathbf{Q^{X_t}}$ and $\mathbf{Q^{Y_t}}$; \\
\textbf{Output}: Approximated $\varepsilon_l$ and $\varepsilon_u$
\end{flushleft}
\begin{algorithmic}[1]
    \State {$\mathbf{\varepsilon_u} \gets \text{[ ]}$}
    \State {$\mathbf{\varepsilon_l} \gets \text{[ ]}$}
    \For {$t = 0, 0.001, 0.002, ..., 0.999, 1.000\}$}
        \State{$\tilde{F}_X \gets \text{PCHIP}(\mathbf{Q^{X_t}},\mathbf{q})$}
        \State{$\tilde{F}_Y \gets \text{PCHIP}(\mathbf{Q^{Y_t}},\mathbf{q})$}
        \For{$i = 1,2,...,|\mathbf{q}|$}
            \If{$\tilde{F}_X(i)$ > $\tilde{F}_X(i)$}
            \State{$\tilde{F}_Y,\tilde{F}_X$ = $\tilde{F}_X,\tilde{F}_Y$}
            \EndIf
            \State{$\mathbf{\tilde{\varepsilon}}.\text{append}(\tilde{F}_Y(i) - \tilde{F}_X(i))$}
        \EndFor
    \EndFor
    \State{$\varepsilon_l \gets \text{max}\{\mathbf{\tilde{\varepsilon}}\}$}
    \State{$\varepsilon_u \gets -\text{min}\{\mathbf{\tilde{\varepsilon}}\}$}
\end{algorithmic}
\end{algorithm}

\subsection{Approximating $Z$}
The approach of using interpolated quantiles to approximate $\tilde{\epsilon}_l$ and $\tilde{\epsilon}_u$ can be used to approximate the bounds for $Z = x_i - y_j$ for $\forall x_i \match y_j$ as well. Again, in scenarios where the CDF for $X$ and $Y$ are ``smooth'', we can use the PCHIP interpolator to obtain a modified version of algorithm \ref{alg:ISMB}. For \textit{zero-violation models}, recall $Z^L \to Z^U = Z$ as the number of quantiles approaches infinity uniformly over the range $[0,1]$. Thus using PCHIP to approximate $Z^U_t$ and $Z^L_t$ would result in $Z^L_t = Z^U_t$, as reflected in the bottom right subplot of \ref{fig:Z_different_e}. The implementation details are given in algorithm \ref{alg:AIMB}. To approximate a \textit{zero-violation model}, all that's needed is to apply algorithm \ref{alg:AIMB} with $\epsilon_u = \epsilon_l = 0$.

\begin{algorithm}
\begin{flushleft}
\textbf{Input}: Quantile labels $\mathbf{q}$, quantiles $\mathbf{Q^{X}}$ and $\mathbf{Q^{Y}}$\\
\textbf{Output}: Discrete Random Variables $Z^U$ and $Z^L$
\end{flushleft}
\caption{Approximated Iterated Model Bound}\label{alg:AIMB}
    \begin{algorithmic}[1]
        \State{$\mathbf{Z_t^U} \gets [\text{ } ];\text{ } \mathbf{Z^L_t} \gets [\text{ }]$}
        \State{$\tilde{Q}^{X} \gets \text{PCHIP}^{-1}(\mathbf{q},\mathbf{Q^{X}})$}
        \State{$\tilde{Q}^{Y} \gets \text{PCHIP}^{-1}(\mathbf{q},\mathbf{Q^{Y}})$}
        \For{$k = 1,2,...,10^5$}
            \State{Sample $\frac{i}{N}$ from Uniform$(0,1)$}
            \State{$q_l \gets \frac{i}{N} - \epsilon_l ;\text{ } q_u \gets \frac{i}{N} + \epsilon_u$}
            \State{Append $\tilde{Q}^{X}(q_u) - \tilde{Q}^{Y}(q_l)$ to $\mathbf{Z^U_t}$}
            \State{Append $\tilde{Q}^{X}(q_l) - \tilde{Q}^{Y}(q_u)$ to $\mathbf{Z^L_t}$}
        \EndFor
    \Return {$\mathbf{Z^U_t}, \mathbf{Z^U_t}$}
    \end{algorithmic}
\end{algorithm}


\section{Experiments}
We evaluate our proposed methods on the multi-model prediction results for round 9 and round 11, where round 9 has $t_{\text{app}} > t_0$ and round has $t_{\text{app}} = t_0$; that is, the divergence of scenarios in round 9 is delayed while round 11 is immediate, and each increment in the timestamp $t$ represents a week. We first use round 9 to demonstrate the estimation and approximation of $\epsilon$. As it is only possible to estimate $\epsilon$ for a scenario with $t > t_{\text{app}}$. Then, we devise a reasonable guess of $\epsilon$ and evaluate the difference of scenarios with that $\epsilon$.

\subsection{Estimation and Approximation of $\epsilon$}

    \begin{tabular}{c|c|c|c|c}
    Model & $\tilde{\epsilon}_l$ & $\tilde{\epsilon}_u$ & $\varepsilon_l$ & $\varepsilon_u$ \\
    \hline
    USC SIkJalpha \cite{sikjalpha}& \textbf{0.05} & \textbf{0.1} & 0.013 & 0.013 \\
    Ensemble & \textbf{0.05} & \textbf{0.1} & 0.038 & 0.037 \\
    Ensemble LOP & \textbf{0.05} & \textbf{0.1} & 0.025 & 0.025 \\
    Ensemble LOP Untrimmed & \textbf{0.05} & \textbf{0.1} & 0.02 & 0.025 \\
    JHUAPL Bucky \cite{bucky}& \textbf{0.05} & \textbf{0.05} & 0.00 & 0.00 \\
    MOBS NEU-GLEAM COVID \cite{mobs}& \textbf{0.05} & \textbf{0.1} & 0.05 & 0.05 \\
    \end{tabular}
    \makeatletter\def\@captype{table}\makeatother
    \caption{Summary of $\epsilon$ for round 9 scenario `B' and `A' cumulative cases. For each row, the largest of $\tilde{\epsilon}$ and $\varepsilon$ are marked in bold}\label{table:ep}

In round 9 $t_{\text{app}} = 4$. In other words, with prediction starting at $t_0 = 0$, there is a total of 4 weeks' data available before $t_{\text{app}}$ for obtaining $\epsilon$. We demonstrate the estimation and approximation of both $\epsilon_u$ and $\epsilon_l$ for a list of models participating in the multi-model COVID-19 prediction effort. In figure \ref{fig:epsilon_models}, we choose scenarios `A' and `B' as the scenarios of interest; namely, scenario `A' refers to adopting childhood vaccination and no new COVID-19 variant, while scenario `B' refers to no childhood vaccination and no new COVID-19 variant, all of which are for cases cumulative data. We observe from the left subplot that the estimated epsilon violation $\tilde{\epsilon}_l$ and $\tilde{\epsilon}_u$ is small and clusters for all the models being examined. For the approximations $\varepsilon_l$ and $\epsilon_u$, a linear relationship is observed between the two values for any given \textit{model}. The numerical results are given in table \ref{table:ep}. Notice that for all of the \textit{models} examined, $\varepsilon \leq \tilde{\epsilon}$ as expected.

\begin{figure}
    \centering
    \includegraphics[width=1\linewidth]{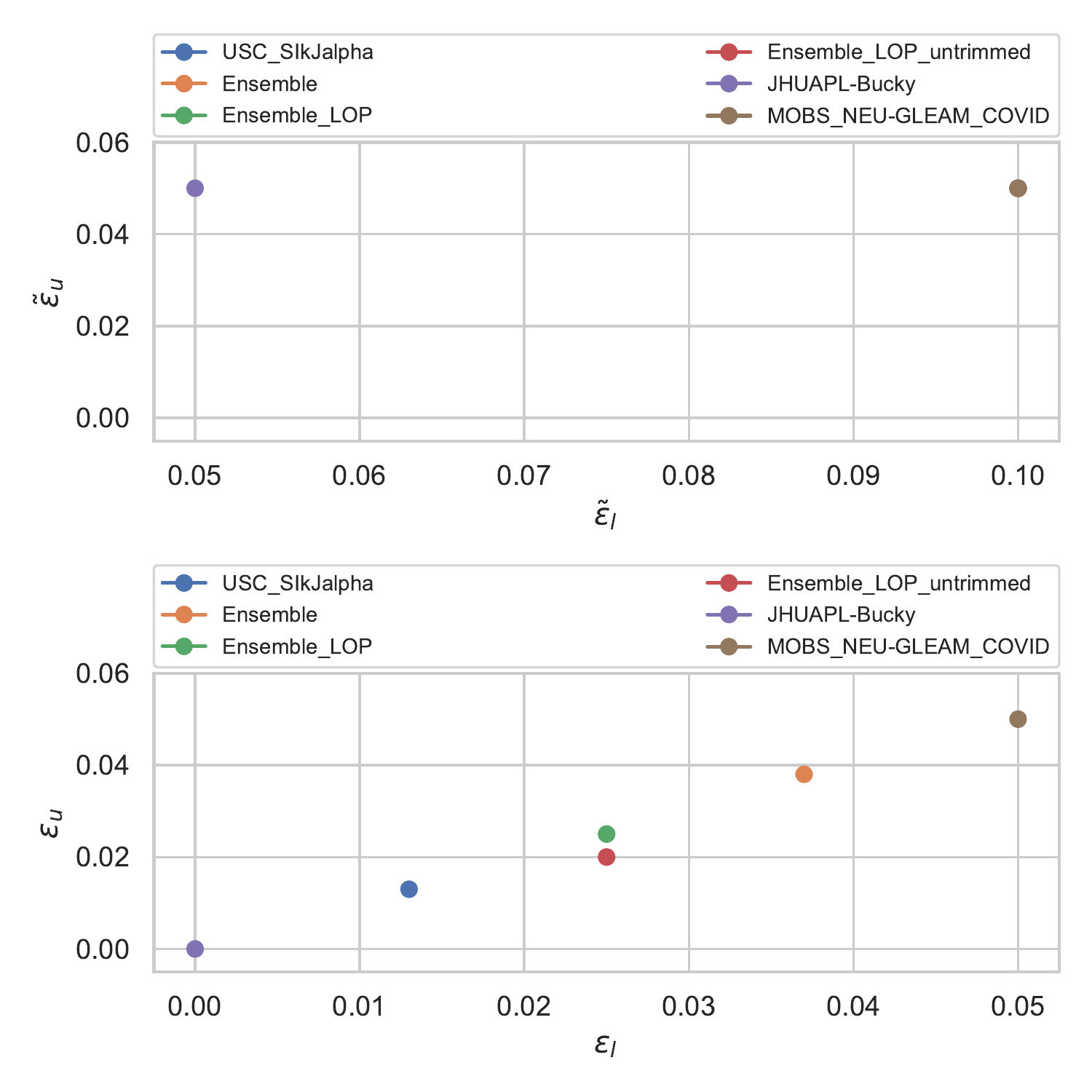}
    \caption{$\tilde{\epsilon}$ and $\varepsilon$ plots for various models for round 9 scenarios `A' and `B. For the top subplot, $\tilde{\epsilon}_l$ and $\tilde{\epsilon}_u$ are clustered at two coordinates on the graph. This is understandable as $\tilde{\epsilon}_u$ and $\tilde{\epsilon}_l$ are always greater or equal to the smaller difference in quantiles; that is, for models with violations smaller than the smallest difference in quantiles, the estimated violation gets rounded up. For approximated $\varepsilon$, the bottom subplot, clustering is not observed as $\varepsilon$ are calculated from interpolated quantiles where the smallest difference in the quantiles is close to 0.}
    \label{fig:epsilon_models}
\end{figure}

\subsubsection{Distribution of $\tilde{\epsilon}$ and $\varepsilon$}
Since we obtain $\tilde{\epsilon}$ and $\varepsilon$ via maximization on $t \leq t_{\text{app}}$, it is worthwhile to examine the distribution of both $\tilde{\epsilon}$ and $\varepsilon$ versus $t$ to make sure the maximums obtained are not outliers. Again, we use the same round and environment settings as before. From figure \ref{fig:e_dist}, the value of $\tilde{\epsilon}_l$ stays constant in the $4$ weeks that are examined, and although slight variations in $\tilde{\epsilon}_u$ is present, no significant outlier exist to have a significant impact on the overall $\tilde{\epsilon}_u$.

\begin{figure}
    \centering
    \includegraphics[width=1\linewidth]{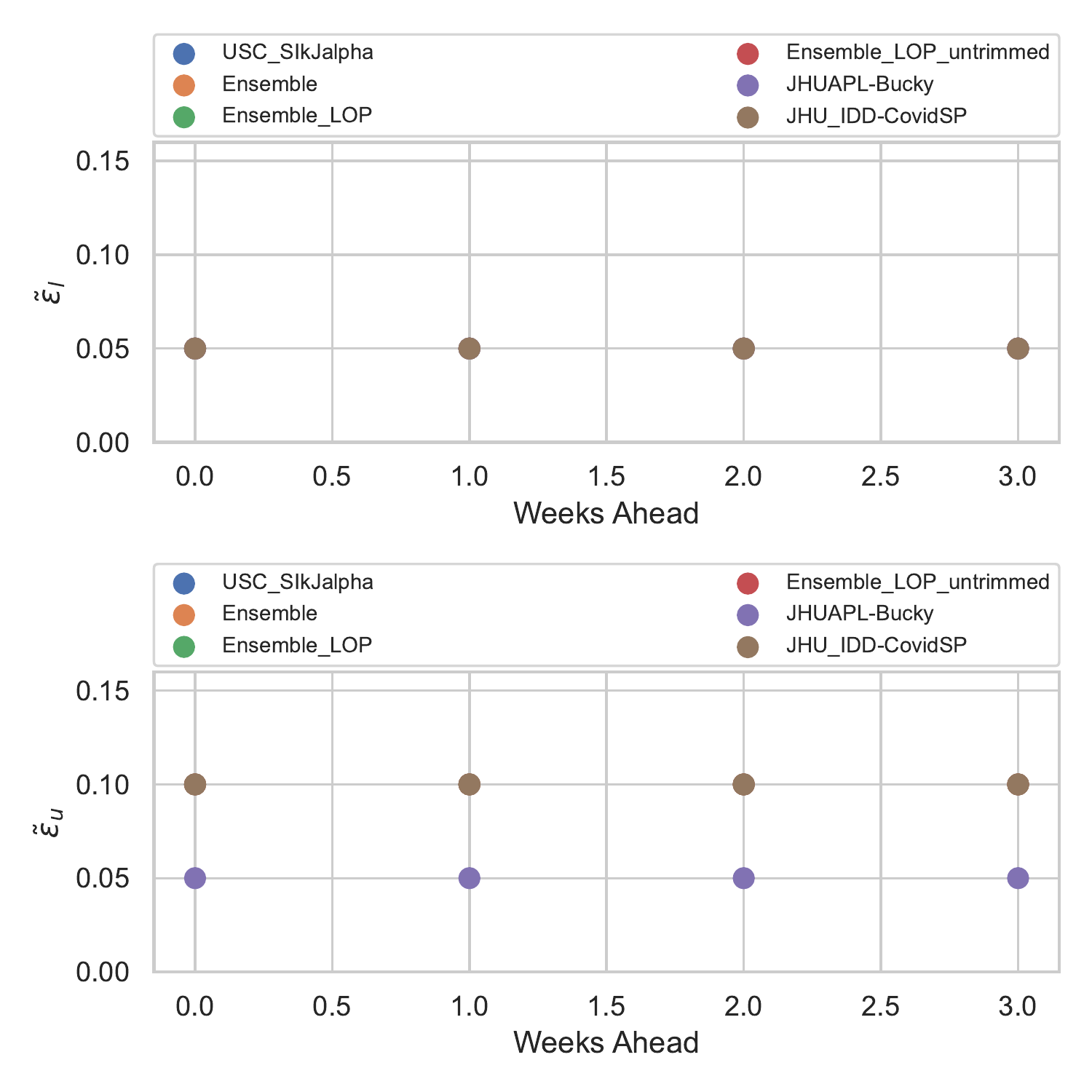}
    \caption{Distribution of $\tilde{\epsilon}$ over time for round 9 cases cumulative scenarios `A' and `B'. As observed, since both $\tilde{\epsilon}$ and $\varepsilon$ are maximized over time, the estimates for $\epsilon$ are reasonable as no significant variation is observed over time.}
    \label{fig:e_dist}
\end{figure}

For the distribution of $\varepsilon$ in figure \ref{fig:vare_dist}, the clustering effect is significantly reduced as compared to the non-interpolated estimate of $\epsilon$ in figure \ref{fig:e_dist}. From observation, although $\varepsilon$ has a higher variance across weeks for any particular model, the overall distribution remains fairly uniform and no significant outlier exists for all of the models examined.

\begin{figure}
    \centering
    \includegraphics[width=1\linewidth]{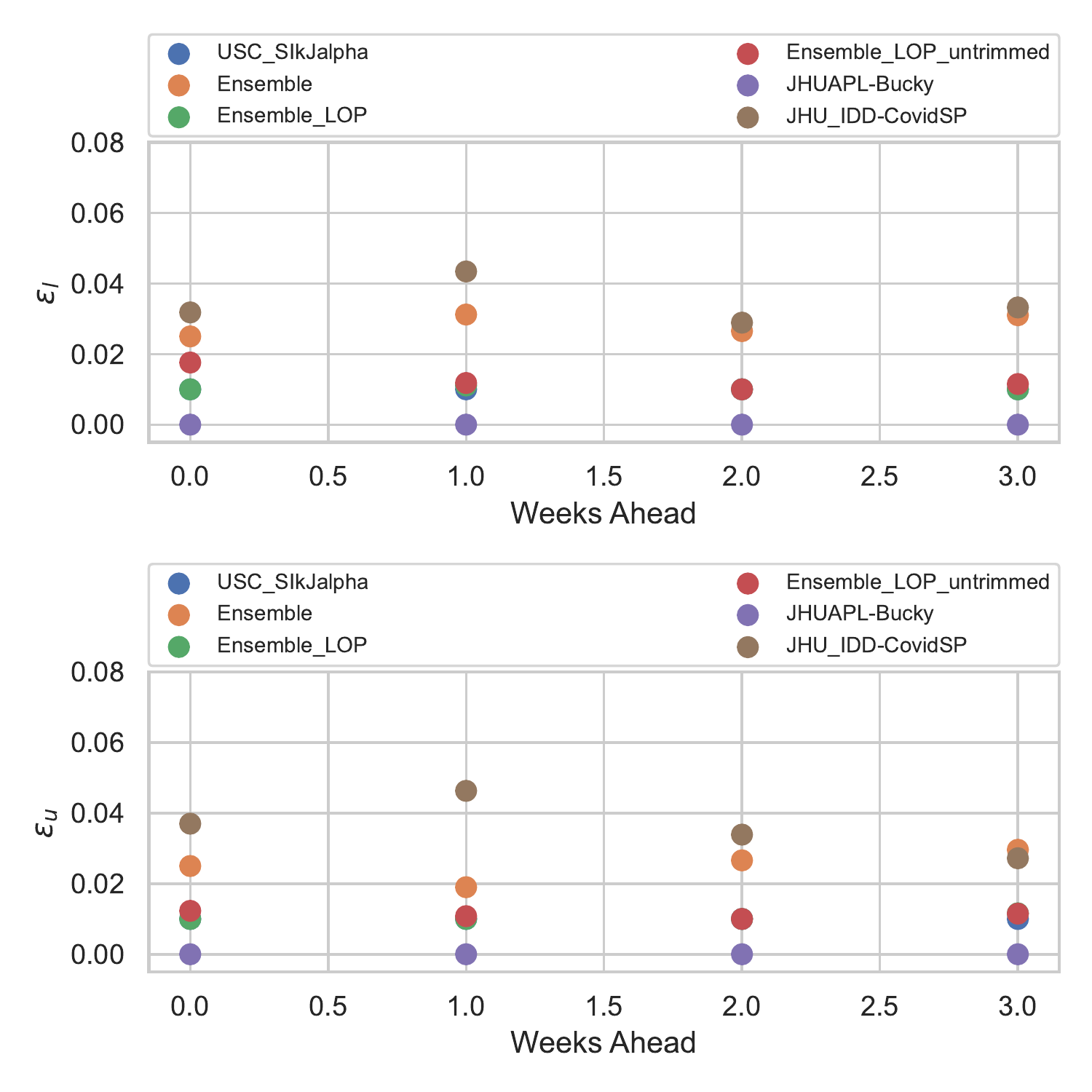}
    \caption{Distribution of $\varepsilon$ over time for round 9 cases cumulative scenarios `A' and `B'. Comparing to the $\tilde{\epsilon}$ in figure \ref{fig:e_dist}, $\varepsilon$ are not lower-bounded by the smallest difference in quantiles. Therefore, as expected, more variation is observed across time.}    \label{fig:vare_dist}
\end{figure}


\subsection{Confidence Intervals of Z}
After obtaining an estimated $\epsilon$ and examining its distribution in detail, we now attempt to use algorithm \ref{alg:ISMB} and algorithm \ref{alg:AIMB} to bound the difference of cumulative cases for round 11 scenario `B' $-$ `A'. Specifically, scenario `B' refers to ``Optimistic severity and high transmissibility increase", while `A' refers to ``Optimistic severity and low transmissibility increase", with both scenarios starting at 2021-12-21. Unlike round 9, round 11 has $t_0 = t_{\text{app}}$, and we therefore cannot estimate nor approximate $\epsilon$ with the algorithm proposed. Yet, we should be able to make reasonable speculation of $\epsilon$ from the previous plots of $\epsilon$ distributions. Instead of taking a conclusive guess, we demonstrate the effect of different $\alpha$ on the resulting $\alpha$-confidence interval. In particular, we use $\alpha = 0.8$. The result from both exact (algorithm \ref{alg:ISMB}) $Z$ and estimated (algorithm \ref{alg:AIMB}) $\alpha$-confidence interval for $Z$ are plotted side by side in figure \ref{fig:Z_different_e}.

\begin{figure}
    \centering
    \includegraphics[width=1\linewidth]{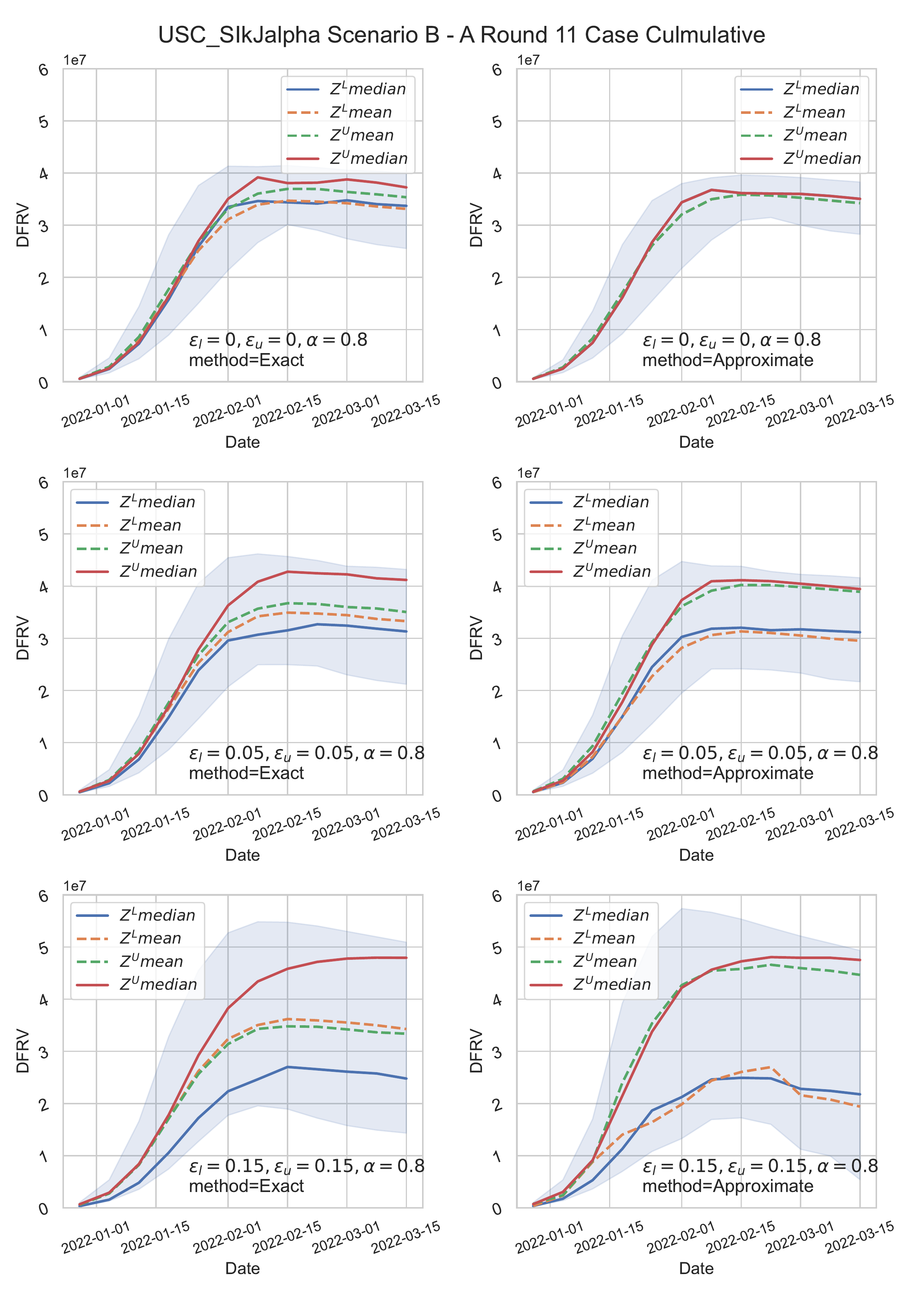}
    \caption{0.8-confidence interval for scenario B - A cases cumulative round 11 for the model USC-SikJ$\alpha$. The left column of subplots results from algorithm \ref{alg:ISMB} and the right column are the approximated bounds via algorithm \ref{alg:AIMB}. This plot demonstrates the effects of varying levels of violation have on the resulting bound.}
    \label{fig:Z_different_e}
\end{figure}

Intuitively, the effect of simultaneously increasing $\epsilon_u$ and $\epsilon_l$ has a two-sided effect on the $\alpha$-confidence interval obtained for some fixed alpha. When $\epsilon_l = \epsilon_u = 0$ for the approximated method (i.e interpolating the quantiles), it has been previously shown that $Z^U$ converges to $Z^L$; in other words, the medians and mean of both random variables collide. This can be seen from the bottom left subplot of figure \ref{fig:Z_different_e}. 

From the plots, we are provided with a way to interpret the results generated by the \textit{model} ``USC SIKJalpha": With at least 0.8 probability, the difference in the projections of highly transmissible and low transmissible variants of COVID-19 is going to fall within the blue region ascribed in figure \ref{fig:Z_different_e}.

\section{Conclusion}
The goal of generating scenario projections based on a future decision is to assess the impact of that decision. In practice, however, as the joint distribution of decisions is unknown, obtaining tight bounds from probabilistic projections is often infeasible without significant changes to the model. Without doing additional work on the model, the methods proposed can find arbitrary $\alpha$-confidence intervals for scenario differences, under some assumptions. The second half of the paper is dedicated to relaxing the assumptions by incorporating a quantitative measure $\epsilon$ on the degree of violation of the assumption. Finally, a method is proposed to reliably approximate the difference in the probabilistic projections of scenarios.  

\section*{Acknowledgement}
This work was supported by the  Centers for Disease Control and Prevention and the National Science Foundation under the awards no. 2135784 and 2223933. Any opinions, findings, and conclusions or recommendations expressed in this material are those of the author and do not necessarily reflect the views of the National Science Foundation or the Center for Disease Control and Prevention.

\bibliographystyle{plainnat}
\bibliography{references}

\appendix
\section{Appendix}
The code, data, and instructions to reproduce the results and visualizations associated with this work can be found in \href{https://github.com/ULY-SS3S/Scenario_Interpretation}{Github}\footnote{\url{https://github.com/ULY-SS3S/Scenario_Interpretation}}. 

\begin{proof}\label{lem:1proof} (Lemma \ref{lem:1})
Consider the \textit{matching} $x_i \match y_j$. Without loss of generality, assume $i>j$, then since each $x$ and $y$ has a one-to-one mapping, there must exist an $i' < i$ and $j' > j$ such that $x_{i'}$ is matched with $y_{j'}$. Otherwise, by the Pigeon-Hole principle, more than one $x$ would need to be matched with one $y$, contradicting the one-to-one mapping of $x$ and $y$.

Consider the changes in latent variables that changed the pairs $(x_i,y_j)$ to $(x_{i'},y_{j'})$. This change decreased $x_i$ to $x_{i'}$ while increased $y_j$ to $y_{j'}$, violating assumption \ref{Assumption1}.
\end{proof}

\begin{proof}\label{lem:2proof}  (Lemma \ref{lem:2})
We abuse the notation for $\min$ and $\max$ to be taken on $t$. This means to take the extremum on for $x_i \match y_j$ on $U(t) \sforall t$.
Since the \textit{model} is \textit{well-behaved}, for $x_i \match y_j$ we have 
\begin{align}
    \epsilon_u &\vcentcolon= -\min_{i,j,t}\{ F_Y(y_j) - F_X(x_i)\} \\
    &= -\min_{i,j,t < t_{\text{app}}}\{ F_Y(y_j) - F_X(x_i)\} \label{eqn:16}
    \end{align}
since $F(x_i)$ is monotonic in $x_i$,
\begin{equation}
F_X(Q_{q_\beta}) \geq F_X(x_i), \text{ where } \beta = \min_{k\geq i}\{Q^Y_{q_k} \geq x_i\}
\end{equation}
it thus follows from \ref{eqn:16} that
\begin{align}
    \epsilon_u &\leq -\min_{j,t<t_{\text{app}}}\{F_Y(y_j) - F_X(Q_{q_\beta})\} ,\text{ where } \beta = \min_{k\geq i}\{Q^Y_{q_k} \geq x_i\}\\
    &= -\min_{j,t<t_\text{app}}\{F_Y(y_j) - q_{\beta}\}, \text{ where } \beta = \min_{k\geq i}\{Q^Y_{q_k} \geq x_i\} \label{eqn:19}
\end{align}
Since $\forall x_i, y_j \text{ }\exists q_{i-1},q_{i}$ such that
\begin{equation}
    q_{i-1} \leq F_Y(y_j) \leq q_{i}
\end{equation}
That is, from lemma \ref{lem:1}, $Q^Y_{q_{i-1}} \leq y_j = x_i \leq Q^Y_{q_i} \leq Q^X_{q_i}$. Then, from \ref{eqn:19} we have
\begin{align}
    \epsilon_u &\leq -\min_{i-1,t<t_{\text{app}}}\{q_{i-1} - q_\beta\}, \text{ where } \beta = \min_{k\geq i}\{Q_k^Y \geq Q_{q_i}^X\}\\
    &= -\min_{i,t<t_{\text{app}}}\{q_{i} - q_\beta\}, \text{ where } \beta = \min_{k\geq i + 1}\{Q^Y_k \geq Q^X_{i+1}\}\\
    &= \tilde{\epsilon}_u
\end{align}
as desired. Similarly, 
\begin{align}
    \epsilon_l &\vcentcolon= \max_{i,j,t}\{ F_Y(y_j) - F_X(x_i)\} \\
    &= \max_{i,j,t < t_{\text{app}}}\{ F_Y(y_j) - F_X(x_i)\} \label{eqn:26}
\end{align}
since $F(x_i)$ is monotonic in $x_i$, 
\begin{align}
    \epsilon_u &\leq \max_{j,t\leq t_{\text{app}}}\{F_Y(y_j) - q_\alpha\} ,\text{ where } \alpha = \max_{k\leq i-1}\{Q^X_{q_k} \geq x_i\}\\
    \intertext{applying lemma \ref{lem:1},}
    &= \max_{j,t\leq t_\text{app}}\{F_Y(y_j) - q_{\alpha}\}, \text{ where } \alpha = \max_{k \leq i-1}\{Q^X_{q_k} \geq y_j\}\label{eqn:27}
\end{align}
Since $\forall x_i, y_j \text{ }\exists q_{i-1},q_{i}$ such that
\begin{equation}
    q_{i-1} \leq F_Y(y_j) \leq q_{i}
\end{equation}
that is, from the assumption $Q^X_i \geq Q^Y_i$,
\begin{equation}
    Q_i^Y \geq y_j = x_i \geq Q_{q{i-1}}^Y \geq Q_{q_{i-1}}^X
\end{equation}
then, from \ref{eqn:27} we have
\begin{align}
    \epsilon_u &\leq -\min_{j,t\leq t_{\text{app}}}\{q_{i} - q_\alpha\}, \text{ where } \alpha = \min_{k \leq i-1}\{Q_k^Y \geq Q_{j-1}^X\} \\
    &= \tilde{\epsilon}_u   
\end{align}
as desired.
\end{proof}
\end{document}